\documentclass{article}

\usepackage{amsfonts}
\usepackage{amssymb,epic,eepic}
\usepackage{amsmath}
\usepackage{dcolumn}
\usepackage{bm}
\usepackage{bbm}
\usepackage{verbatim}
\usepackage{color}
\usepackage{amsthm}
\newcommand{\hr}{{\mathcal H}}
\newcommand{\cn}{{\mathcal N }}

\newcommand{\cs}{{\mathcal S}}

\newcommand{\cc}{{\mathbb C}}

\newcommand{\nn}{{\mathbb N}}

\newcommand{\eps}{{\varepsilon}}        %%%%%%%%%%%%%%%%%%%%%%%%%%%%%%%%%%%
\newcommand{\bS}{\mathbf S}
\newcommand{\bX}{\mathbf X}

\newcommand{\eins}{{\mathbbm{1}}}
\newcommand{\BIGOP}[1]
{
\mathop{\mathchoice%
{\raise-0.22em\hbox{\Large $#1$}}%
{\raise-0.05em\hbox{\large $#1$}}{\hbox{\large $#1$}}{#1}}}

\newcommand{\BIGboxplus}{\mathop{\mathchoice%
{\raise-0.35em\hbox{\huge $\boxplus$}}%
{\raise-0.15em\hbox{\Large $\boxplus$}}{\hbox{\large $\boxplus$}}{\boxplus}}}

\newtheorem{theorem}{Theorem}

\newtheorem{lemma}{Lemma}

\newtheorem{remark}{Remark}

\newcommand{\tr}{\mathrm{tr}}
\newcommand{\supp}{\mathrm{supp}}

\DeclareMathOperator{\conv}{conv}
\DeclareMathOperator{\aff}{aff}
\DeclareMathOperator{\ri}{ri}
\DeclareMathOperator{\rebd}{rebd}

\DeclareMathOperator{\linspan}{span}

\begin{document}
\title{Hypothesis Testing on Invariant Subspaces of the Symmetric Group, Part I -\\Quantum Sanov's Theorem and Arbitrarily Varying Sources}
\author{J. N\"otzel \\
\scriptsize{Electronic address: janis.noetzel@tum.de}
\vspace{0.2cm}\\
{\footnotesize Lehrstuhl f\"ur Theoretische Informationstechnik, Technische Universit\"at M\"unchen,}\\
{\footnotesize 80290 M\"unchen, Germany}
}
\maketitle

\begin{abstract}
We report a proof of the quantum Sanov Theorem by elementary application of basic facts about representations of the symmetric group, together with a complete characterization of the optimal error exponent in a situation where the null hypothesis is given by an arbitrarily varying quantum source instead. Our approach differs from previous ones in two points: First, it supports a reasoning inspired by the ``method of types''. Second, the measurement scheme we propose to distinguish the two alternatives not only does that job asymptotically perfect, but also yields additional information about the null hypothesis. An example of that is given. The measurement is composed of projections onto permutation-invariant subspaces, thus providing a direct link between one of the most basic tasks in quantum information on the one hand side and fundamental objects in representation theory on the other.\\
We additionally connect to representation theory by proving a relation between Kostka numbers and quantum states, and to state estimation via a generalization of a well-known spectral estimation theorem to non-i.i.d. sequences.
\end{abstract}
%-----------------------------------------------------------------------------------
\begin{section}{Introduction}
The importance of the symmetric group in quantum Shannon information is inherently connected to the central idea of using an information carrying or transmitting system several times in order to defeat noise. This approach naturally introduces an action of the symmetric group into any model within that theory. On these pages, we will therefore put a clear focus on representations of $S_n$ while almost completely ignoring its commutant which is, in our case, e.g. a representation of the unitary group.\\
We will connect two fundamental concepts: hypothesis testing and invariant subspaces of the symmetric group. This also leads to new insights concerning the notion of a quantum method of types.\\
We will now briefly introduce the two concepts, starting with hypothesis testing. More specifically, we will consider asymmetric hypothesis testing, in the setting of Stein's Lemma.
\\\\
Stein's Lemma is one of the fundamental statements in information theory. It gives precise bounds on the probability of correctly identifying the state of a system at hand in situations where it is guaranteed that the actual state is chosen from a set of two states (I and II). The underlying assumption is that infinitely many copies of the system in the exact same state can be prepared, and then measured jointly. The measurement should then identify the true state of the system. The situation is not symmetric: While it is desired to identify state I only with a constant nonnegative probability, the probability of correctly identifying state II goes to one asymptotically fast, and the optimal exponent is quantified in Stein's Lemma.\\
In case the system at hand is described by quantum theory, the first proof of the direct of the corresponding quantum Stein's Lemma was given by Hiai and Petz in \cite{hiai-petz}. Soon after, Ogawa and Nagaoka \cite{ogawa-nagaoka} completed the proof by showing the converse. Since then, there have been different variants of the proof and extensions to more general scenarios: Hayashi \cite{hayashi-representation-theory} proved the quantum Stein's Lemma using representation theory of the special linear group and further developed the idea of approximating the quantum relative entropy by the classical relative entropy between the distribution of the measurement results for specific measurements. It remains to be seen whether our result can be treated similarly. Our work differs from his approach foremost in the very clear and explicit choice of parameters that we use to describe our measurements. We hope that this will increase applicability of methods from representation theory to quantum Shannon information theory. Hayashi also investigated the problem from a different point of view in \cite{hayashi01} and \cite{hayashi02}, namely with an information-spectrum approach. This work also further clarifies the connection between statistics of measurements performed on asymptotically many copies of quantum states and the relative entropy of the respective states. Meanwhile, \cite{nagaoka} gave a much simpler proof of the converse. Later, Ogawa and Hayashi together investigated the interplay between the different rates of convergence of the two errors. The generaliztaion to the ergodic case was done by Bjelakovi\c{c} and Siegmund-Schultze in \cite{bjelakovic-siegmund_schultze-ergodic}, and later generalized to the case of a stationary (instead of i.i.d.) hypothesis by Bjelakovi\c{c}, Deuschel, Kr\"uger, Seiler, Siegmund-Schultze, Szkola in \cite{bjelakovic-deuschel-krueger-seiler-siegmund_schultze-szkola-typical-support-and-sanov}. The proof of the quantum Sanov theorem was first given in the i.i.d. case by Hayashi (Theorem 2 in \cite{hayashi-representation-theory}), only without attributing the specific name to it. It was then independently reproven by Bjelakovi\c{c}, Deuschel, Kr\"uger, Seiler, Siegmund-Schultze, Szkola in \cite{bjelakovic-deuschel-krueger-seiler-siegmund_schultze-szkola-sanov} and extended to the stationary case in \cite{bjelakovic-deuschel-krueger-seiler-siegmund_schultze-szkola-typical-support-and-sanov}.\\
Another approach is that of Audenaert, Nussbaum, Szkola and Verstraete \cite{audenaert-nussbaum-szkola-verstraete} who were concerned with asymmetric hypotheses testing, their results contain the quantum Sanov theorem as a corollary. Their results got extended to correlated states by Hiai, Mosonyi and Ogawa in \cite{hiai-mosonyi-ogawa}.\\
It was noticed in several of these papers that the choice of the optimal POVM depends on the reference state, and this turns out to be no different in our approach.\\
A complete estimation scheme that is based on representation theory of the symmetric group has also been described in \cite{keyl}, including its large deviation behavior.\\
Lately, an important generalizations of Stein's Lemma was given for the case where the alternative hypothesis consists of a sequence of sets of states satisfying certain stability properties results by Brandao and Plenio \cite{brandao-plenio}, and their results were extended to the case of restricted measurements in \cite{brandao-harrow-lee-peres}.
\\\\
Let us now state the most basic task that is to be performed in hypothesis testing: Given two states $\rho,\sigma\in\cs(\mathbb C^d)$ and $\nu\in(0,1)$, determine
\begin{align}
\lim_{n\to\infty}\frac{1}{n}\log(\min_{0\leq P\leq\eins^{\otimes n}}\{\tr\{P\sigma^{\otimes n}\}:\tr\{P\rho^{\otimes n}\}\geq1-\nu\}).
\end{align}
The quantity $\min_{0\leq P\leq\eins^{\otimes n}}\{\tr\{P\sigma^{\otimes n}\}:\tr\{P\rho^{\otimes n}\}\geq1-\nu\}$ is usually abbreviated as $\beta_{\nu,n}(\rho,\sigma)$ and named the 'type two error'. The statement of the quantum Stein's Lemma is then:
\begin{align}
\forall\ \nu\in(0,1):\qquad\lim_{n\to\infty}\frac{1}{n}\log(\beta_{\nu,n}(\rho,\sigma))=-D(\rho\|\sigma).
\end{align}
In case we want to distinguish more subtle hypotheses, we have to adjust this definition. In our case, it will be sufficient to look at the following: For a state $\sigma\in\cs(\mathbb C^d)$, $\nu\in(0,1)$ and a sequence $\Delta:=(\mathfrak S_n)_{n\in\nn}$ of sets such that $\mathfrak S_n\subset\cs((\mathbb C^d)^{\otimes n})$, we wish to investigate the quantity
\begin{align}
\beta_{n,\nu}(\Delta,\sigma):=\min_{0\leq P\leq\eins^{\otimes n}}\{\tr\{P\sigma^{\otimes n}\}:\inf_{\omega\in\mathfrak S_n}\tr\{P\omega\}\geq1-\nu\}.
\end{align}
Our main result (which we state in a much more explicit way later due to intended applications to other problems) can then be formulated as follows.
\begin{theorem}\label{thm:mainresult-short-form}
First, if $\Delta=(\{\rho^{\otimes n}\}_{\rho\in\mathfrak S})_{n\in\nn}$ for some set $\mathfrak S\subset\cs(\mathbb C^d)$, then
\begin{align}
\forall\nu\in(0,1),\qquad\lim_{n\to\infty}\frac{1}{n}\log(\beta_{n,\nu}(\Delta,\sigma))=-\inf_{\rho\in\mathfrak S}D(\rho\|\sigma).
\end{align}
Second, if $\Delta=(\{\rho_{s^n}\}_{s^n\in\bS^n})_{n\in\nn}$, where $\rho_{s^n}:=\otimes_{i=1}^n\rho_{s_i}$ and each $\rho_{s_i}\in\mathfrak S$ for some set $\mathfrak S\subset\cs(\mathbb C^d)$, then
\begin{align}
\forall\nu\in(0,1),\qquad\lim_{n\to\infty}\frac{1}{n}\log(\beta_{n,\nu}(\Delta,\sigma))=-\inf_{\rho\in\conv(\mathfrak S)}D(\rho\|\sigma).
\end{align}
\end{theorem}
\begin{remark}
The optimal tests are, in both settings, essentially given by projections onto suitably chosen invariant subspaces of the symmetric group.
\end{remark}
The first result in Theorem \ref{thm:mainresult-short-form} is usually referred to as 'Quantum Sanov Theorem'. We will call the second scenario hypothesis testing for an arbitrarily varying quantum source (sometimes abbreviated as AVQS in the following), where the terminology 'arbitrarily varying source' is taken from the classical literature and refers to the sequence $\Delta$ in that setting. This model has been investigated first in \cite{fangwei-shiyi} for the case of two sets of classical probability distributions over a set. Their proof is based on the 'method of types'.
\\\\
The 'method of types' is best described by Csisz\'{a}r in \cite{csiszar-types}. The method is based on a study of the relevant quantities in a given communication scenario on \emph{frequency typical sets}, meaning sets of words of length $n$ that are each composed of only a finite number of symbols and each symbol occurs a constant number of times in each word. It is clear that such sets are invariant under a joint permutation of the symbols in each word. The success of the method can be read off from the fact that it ultimately ended up having its own name. Approaching a given Shannon-information theoretic problem by looking at types can be considered opposed to looking at sets on which one controls the empirical averages of the functions that are of interest in a given problem.\\
The only stringent attempt to generalize the method of types to the quantum setting that is known to the author is presented in \cite{harrow-diss} and \cite{bacon-chuang-harrow}, where the group theoretic structure is presented in great detail. Certainly, also \cite{christandl-thesis} is inspired by that idea.\\
The methods that were further successfully applied often only use spectral estimation, which already reveals important connections between representation theory and quantum states \cite{christandl-mitchison}, \cite{christandl-harrow}, \cite{christandl-horn_and_LW_coefficients}. They are well suited to derive e.g. estimates of the spectra of states through the basic estimate described in \cite{thm2} or entropic inequalities. On a more general level, recent results \cite{schilling-christandl-gross,christandl-sahinoglu-walter,walter-doran-gross-christandl} show that there is an enormously fruitful connection between representation theory and the structure of multiparty states.\\
Unfortunately, the quantum method of types as such was not applied or developed any further and could thus not even be used to give proofs for such basic tasks as hypothesis testing, which are both simple and at the very heart of quantum Shannon information theory. At this point, we clearly distinguish between applying the structure given by representation theory as was successfully done in \cite{hayashi-representation-theory} and a reasoning in terms of anything comparable to classical types. We thus now make an explicit choice of 'types' that singles out certain invariant subspaces of the symmetric group, and then show that with this choice we can prove our Theorem \ref{thm:mainresult-short-form}. We hope that this will give new thrust to the idea of 'quantum types'.\\
At present, it seems not totally clear what such a theory should look like, the relevant parameters will have to be found by subsequently proving the important theorems in quantum Shannon information theory. For the time being, we shall restrict to considering decompositions $\eins^{\otimes n}=\sum_iP_i$ on $(\mathbb C^d)1 {\otimes n}$ where each $P_i$ is an orthogonal projection onto an irreducible subspace of the natural action of $S_n$ on $(\mathbb C^d)1 {\otimes n}$ and further $\tr\{P_iP_j\}=0$ for $i\neq j$. We then consider the set $\{P_i\}_i$ as a POVM, which is to replace the frequency typical subsets from the classical theory. An outline of our approach will now be given.
\\\\
Despite above mentioned drawbacks, the methods from classical information theory were taken over to the quantum case with great success, e.g. in the form of 'frequency typical subspaces' that arise from choosing a particular orthonormal basis that is labelled by a classical alphabet and then looking at the linear span of all elementary tensor products of $n$ of these basis vectors that are built up from the same number of each of the basis vectors. Closely related is the study of isotypical subspaces in the $n$-fold tensor product of $\mathbb C^d$. Like typical sets in the classical case, both of the above classes of subspaces are permutation invariant.\\
This work combines the two approaches by exploiting the fact that every frequency typical subspace can be decomposed into subspaces that are irreducible under the action of the permutation group.\\
As mentioned already, the action of the commutant of $S_n$ on the frequency typical subspaces will be insignificant in every single one of the proofs given here. Its exact role will therefore also be addressed in future work.
\\\\
The basic idea behind the paper is to look, for a given orthonormal basis $\{e_i\}_{i=1}^d$ and a 'frequency' or 'type' (a nonnegative function $f:\{1,\ldots,d\}\to\nn$ satisfying $\sum_if(i)=n$), at representations of the symmetric group $S_n$ on the frequency typical subspaces
\begin{align*}V_f:=\linspan(\{e_{i_1}\otimes\ldots\otimes e_{i_n}:|\{k:i_k=i\}|=f(i)\ \mathrm{for\ all\ }i\}.\end{align*}
Since these are invariant under permutations, they naturally split up into different isotypical subrepresentations $V_{f,\lambda}$ (where $\lambda$ denote Young Tableaux and some $V_{f,\lambda}$ may not contribute to above decomposition, meaning that $V_{f,\lambda}=\{0\}$ for some pairs $(f,\lambda)$)
\begin{align}V_f=\bigoplus_\lambda V_{f,\lambda}.\end{align}
Given a state $\sigma$ with eigenvalues $t_1\geq t_2\geq\ldots\geq t_d>0$, we may now pick one of its eigenbases for the definition of the $V_f$. It is then straightforward to show that for the orthogonal projections $P_{f,\lambda}$ onto the $V_{f,\lambda}$ the estimate
\begin{align}\tr\{P_{f,\lambda}\sigma^{\otimes n}\}&\approx\dim(F_\lambda)\cdot2^{-n\sum_i\frac{1}{n}f(i)\log t_i}\\
&\approx2^{n(H(\frac{1}{n}\lambda)-\sum_i\frac{1}{n}f(i)\log t_i)}
\end{align}
is valid, where $\dim(F_\lambda)$ is the dimension of the irreducible representation of $S_n$ corresponding to $\lambda$ and $H(\frac{1}{n}\lambda)=-\sum_i\frac{1}{n}\lambda_i\log(\frac{1}{n}\lambda_i)$ is the entropy of the normalized Young tableau. It turns out that, for $\lambda\approx n\cdot\mathrm{spec}(\rho)$ and $f(i)\approx n\cdot\langle e_i,\rho e_i\rangle\ \forall\ i\in[d]$ for some arbitrary second state $\rho$ we get
\begin{align}\tr\{P_{f,\lambda}\sigma^{\otimes n}\}&\approx2^{-nD(\rho\|\sigma)}.\end{align}
It then remains to investigate the behaviour of the functions $(f,\lambda)\mapsto\tr\{P_{f,\lambda}\rho^{\otimes n}\}$ and determine where they attain there maximum for fixed $\rho$. Luckily it turns out that the maximum lies indeed around the values we have already chosen above, and this is sufficient for a proof of the direct part of the hypothesis testing problem in both cases.\\
For sake of completeness, we prove the converse parts in both of our problems by resorting to a reasoning that is based on the method we just outlined above. A crucial ingredient to this part of our proof is the permutation invariance of the sets of states that we are trying to discriminate. This enables one to assume that an arbitrary test is permutation invariant as well, and ultimately this provides the link between an arbitrary test and the performance of tests that are made up from projections onto the $P_{f,\lambda}$.
\\\\
\emph{Outline.} The paper is structured as follows. We first fix some basic notation in Section \ref{sec:Notation} and state a few preliminaries in Section \ref{sec:Definitions}.\\
We then state our results in Section \ref{sec:Main Results}, in a very technical fashion compared to Theorem \ref{thm:mainresult-short-form}. Theorems \ref{theorem:main result} and \ref{theorem:main-result-II} provide a proof for Theorem \ref{thm:mainresult-short-form}, while Theorem \ref{theorem:kostka-numbers} makes the connection to the Kostka numbers and Lemma \ref{lemma:state-estimation} provides a more general formula for (spectral) state estimation via projections onto isotypical projections in the sense of \cite{thm2}. Our additional Theorem \ref{theorem:additional-result} is a straightforward generalization of the robustification Lemma in \cite{ahlswede-avc-with-state-seq-known-to-sender}. In this version, it can be abbreviated as: any POVM element on $\mathcal B((\mathbb C^d)^{\otimes n})$ that detects all elementary products states $\rho^{\otimes n}$ with high probability is close to the identity, whence has poor performance when used as POVM element in a binary test for distinguishing between the two hypotheses 'entangled state' versus 'separable state'. The proof of the theorem also reveals that our method could benefit from further refinement.\\
The rest of the paper is devoted to the proofs of these statements, the proof of Lemma \ref{lemma:state-estimation} being implicit in that of Theorem \ref{theorem:main-result-II}.\\
The proof of Theorem \ref{theorem:additional-result} suggests that there is still some important information to be gained by using POVMs that are made up from irreducible projections, whereas before we used only invariant ones.
\end{section}
%-----------------------------------------------------------------------------------
\begin{section}{\label{sec:Notation}Notation}
All Hilbert spaces are assumed to have finite dimension and are over the field $\cc$. The set of linear operators from $\hr$ to $\hr$ is denoted $\mathcal B(\hr)$. The adjoint of $b\in\mathcal B(\hr)$ is marked by a star and written $b^\ast$.\\
$\cs(\hr)$ is the set of states, i.e. positive semi-definite operators with trace (the trace function on $\mathbb B(\hr)$ is written $\tr$) $1$ acting on the Hilbert space $\hr$. Pure states are given by projections onto one-dimensional subspaces. A vector $x\in\hr$ of length one spanning such a subspace will therefore be referred to as a state vector, the corresponding state will be written
$|x\rangle\langle x|$. For a finite set $\mathbf X$ the notation $\mathfrak{P}(\mathbf X)$ is reserved for the set of probability distributions on $\mathbf X$, and
$|\mathbf X|$ denotes its cardinality. For any $n\in\nn$, we define $\bX^n:=\{(x_1,\ldots,x_n):x_i\in\bX\ \forall i\in\{1,\ldots,n\}\}$, we also write $x^n$ for the elements of $\bX^n$. Given such element, $N(\cdot|x^n)$ denotes its type, and is defined through $N(x|x^n):=|\{i:x_i=x\}|$.\\
For any natural number $L$, we define $[L]$ to be the shortcut for the set $\{1,...,L\}$.\\
The von Neumann entropy of a state $\rho\in\mathcal{S}(\hr)$ is given by
\begin{equation}S(\rho):=-\textrm{tr}(\rho \log\rho),\end{equation}
where $\log(\cdot)$ denotes the base two logarithm which is used throughout the paper.\\
Given two states $\rho,\sigma\in\cs(\mathbb C^d)$, the relative entropy of them is defined as
\begin{align}
D(\rho\|\sigma):=\left\{\begin{array}{l l}\tr\{\rho(\log(\rho)-\log(\sigma)\},&\mathrm{if}\ \supp(\rho)\subset\supp(\sigma),\\\infty,&\mathrm{else} \end{array}\right.
\end{align}
Another way of measuring distance between quantum states is obviously given by using the one-norm, which obeys:
\begin{align}
\|\rho-\sigma\|:=2\max_{0\leq P\leq\eins}\tr\{P(\rho-\sigma)\}
\end{align}
We now fix our notation for representation theoretic objects and state some basic facts.\\
The symbols $\lambda,\mu$ will be used to denote Young frames. The set of Young frames with at most $d\in\nn$ rows and $n\in\nn$ boxes is denoted $\mathbb Y_{d,n}$.\\
For any given $n$, the representation of $S_n$ we will consider is the standard representation on $(\mathbb C^d)^{\otimes n}$ that acts by permuting tensor factors. Throughout, the dimension $d$ of our basic quantum system will remain fixed.\\
The unique complex vector space carrying the irreducible representation of $S_n$ corresponding to a Young Tableau $\lambda$ will be written $F_\lambda$.\\
The multiplicity of an irreducible subspace of our representation corresponding to a Young frame $\lambda$ is denoted $m_{\lambda,n}$, and this quantity can be upper bounded by $m_{\lambda,n}\leq(2n)^{d^2}$ (see \cite{christandl-thesis}).\\
For $\lambda\in\mathbb Y_{d,n}$, $\bar\lambda\in\mathfrak P([d])$ is defined by $\bar\lambda(i):=\lambda_i/n$. If $\rho\in\cs(\mathbb C^d)$ has spectrum $s\in\mathfrak P([d])$ (in case that $\rho$ has degenerate eigenvalues we count them multiple times!), then it
will always be assumed that $s(1)\geq\ldots\geq s(d)$ holds and the distance between a spectrum $s$ and a Young frame $\lambda\in\mathbb Y_{d,n}$ is measured by $\|\bar\lambda-s\|:=\sum_{i=1}^d|\bar\lambda(i)-s(i)|$. The distance between two probability distributions $p,q\in\mathfrak P([d])$ will be measured by $\|p-q\|:=\sum_i|p(i)-q(i)|$.\\
The Kostka numbers $K_{f,\lambda}$ are as defined in e.g. Fulton's book \cite{fulton}, pages 25-26.\\
We now define two important entropic quantities. Given a finite set $\bX$ and two probability distributions $r,s\in\mathfrak P(\bX)$, we define the relative entropy $D(r||s)$ by
\begin{align}
D(r||s):=\left\{\begin{array}{ll}\sum_{x\in \bX}r(x)\log(r(x)/s(x)),&\mathrm{if}\ s\gg r\\ \infty,&\mathrm{else}\end{array}\right.
\end{align}
In case that $D(r||s)=\infty$, for a positive number $a>0$, we use the convention $2^{-aD(r||s)}=0$. The relative entropy is connected to $\|\cdot\|$ by Pinsker's inequality $D(r||s)\leq\alpha\|r-s\|^2$, where $\alpha:=1/2\ln(2)$.\\
The  entropy of $r\in\mathfrak P(\bX)$ is defined by the formula
\begin{align}
H(r):=-\sum_{x\in \bX}r(x)\log(r(x)).
\end{align}
Throughout, we will be having one fixed state $\sigma\in\cs(\mathbb C^d)$ having a (non-unique) decomposition $\sigma=\sum_{i=1}^dt_i|\tilde e_i\rangle\langle\tilde e_i|$ and the pinching of an arbitrary state $\rho\in\cs(\mathbb C^d)$ to the orthonormal basis $\{\tilde e_i\}_{i=1}^d$ will be given by $\sum_{i=1}^d|\tilde e_i\rangle\langle\tilde e_i|\rho|\tilde e_i\rangle\langle\tilde e_i|$ and induces the probability distribution $\tilde r_\rho\in\mathfrak P([d])$ through $\tilde r_\rho(i):=\langle\tilde e_i,\rho\tilde e_i\rangle$. It is important for the understanding of this paper to keep in mind that the equality $D(\rho\|\sigma)=-H(\mathrm{spec}(\rho))-\sum_{i=1}^d\tilde r_\rho(i)\log(t_i)$ holds.\\
We shall also need some notions on convex sets that we borrow from \cite{webster}, and that are suited for dealing with convex sets on finite dimensional normed spaces $(V, ||\cdot||)$ over the field of real or complex numbers. Let $F\subset V$ be convex. $x\in F$ is said to be a relative interior point of $F$ if there is $r>0$ such that $B(x,r)\cap \aff F\subset F$. Here $B(x,r)$ denotes the open ball of radius $r$ with the center $x$ and $\aff F$ stands for the affine hull of $F$. The set of relative interior points of $F$ is called the relative interior of $F$ and is denoted by $\ri F$. The relative boundary of $F$, $\rebd F$, is the set difference between the closure of $F$ and $\ri F$.\\
For a set $A\subset V$ and $\delta\ge 0$ we define the parallel set or the blow-up $(A)_{\delta}$ of $A$ by
\begin{equation}(A)_{\delta}:=\left\{x\in V: ||x-y||\le \delta \textrm{ for some } y\in A  \right\}.  \end{equation}
For a subset $B\subset\mathbb R^n$ or $B\subset\mathbb C^n$ we denote its convex hull by $\conv(B)$.\\
\end{section}
%-----------------------------------------------------------------------------------
\begin{section}{\label{sec:Definitions}Definitions and preliminary results}
Let $n\in\nn$ be fixed for the moment. The most important technical definition for this work is that of frequency-typical subspaces $V_f$ of $(\mathbb C^d)^{\otimes n}$. These arise from choosing a fixed orthonormal basis $\{e_i\}_{i=1}^d$ of $\mathbb C^d$, choosing a frequency $f$ (a function $f:[d]\to\mathbb N$ satisfying $\sum_{i=1}^df(i)=n$), setting $T_f:=\{(i_1,\ldots,i_n):|\{i_k:i_k=j\}|=f(j)\ \forall j\in[d]\}$, and defining
\begin{align}
V_f:=\linspan(\{e_{i_1}\otimes\ldots\otimes e_{i_n}:(i_1,\ldots,i_n)\in T_f\}).
\end{align}
They have been widely used in quantum informatin theory, but share one very nice property that does not seem to have been exploited yet: They are invariant under permutations. From this property it immediately follows that
\begin{align}
V_f=\bigoplus_\lambda V_{f,\lambda},
\end{align}
where each $V_{f,\lambda}$ is just a direct sum of irreducible representations corresponding to $\lambda$ that is contained entirely within $V_f$.
\begin{comment}The multiplicity of $F_\lambda$ within $V_f$ is given by $\dim(V_{f,\lambda})/\dim(F_\lambda)$.\end{comment}
\\
A fundamental representation theoretic quantity is intimately connected to them: The Kostka numbers. In fact, it holds $K_{f,\lambda}=0\ \Leftrightarrow\ V_{f,\lambda}=\{0\}$
\begin{comment}and if $K_{f,\lambda}>0$ then $K_{f,\lambda}=\dim(V_{f,\lambda})/\dim(F_\lambda)$\end{comment}
, both by definition of the Kostka numbers and by application of Young symmetrizers as described in \cite{sternberg}, pages 254-258.\\
Also, we are going to employ the following estimate taken from \cite{csiszar-koerner} (Lemma 2.3), which is valid for all frequencies $f:[d]\to\nn$ that satisfy $\sum_{i=1}^df(i)=n$:
\begin{equation}
\frac{1}{(n+1)^d}2^{nH(\overline{f})}\leq|T_f|\leq2^{nH(\overline{f})}\label{eqn2}
\end{equation}
We will also need Lemma 2.7 from \cite{csiszar-koerner}:
\begin{lemma}\label{lemma1}
If, for $\mathbf A$ a finite alphabet and $p,q\in\mathfrak P(\mathbf A)$ we have $|p-q|\leq\Theta\leq1/2$, then
\begin{equation}
 |H(p)-H(q)|\leq-\Theta\log\frac{\Theta}{|\mathbf A|}.
\end{equation}
\end{lemma}
Another very important estimate is the following one (a derivation can e.g. be found in \cite{noetzel-2ptypicality}):
\begin{equation}
 2^{n(H(\bar\lambda)-\frac{2d^6}{n}\log(2n))}\leq\dim F_\lambda\leq 2^{nH(\bar\lambda)}\qquad (\lambda\in\mathbb Y_{d,n}).\label{eqn4}
\end{equation}
\end{section}
%--------------------------------------------------------------------------------------------------------------------------------------------------
\begin{section}{\label{sec:Main Results}Main Results}
All results are stated with the underlying Hilbert space being $\mathbb C^d$. Our main result is the finding of a sequence of projections that succeeds in detecting each one of a set $\mathfrak S=\{\rho_s\}_{s\in\bS}$ of quantum states with probability going to one, while having lowest possible probability of detecting another state $\sigma\notin\mathfrak S$. We only give the proof in the case of a nonsingular $\sigma$, since the singular case can in dependence of the set $\{\rho_s\}_{s\in\bS}$ either be treated identically by reduction of the dimension of the underlying system or because an optimal test is already given by projecting onto the usual frequency typical subspaces corresponding to $\sigma$. This reasoning applies to the case of the AVQS as well.\\
Our technical formulation of the quantum Sanov theorem for i.i.d. hypotheses reads as follows.
\begin{theorem}\label{theorem:main result} Let $\mathfrak S=\{\rho_s\}_{s\in\bS}\subset\mathcal S(\mathbb C^d)$ and $\sigma=\sum_{i=1}^dt_i|\tilde e_i\rangle\langle\tilde e_i|\in\cs(\mathbb C^d)$ where all $t_i>0$ be given. We will use the abbreviations $r_s:=r_{\rho_s}$ for the spectra of the states of our set and $\tilde r_s:=\tilde r_{\rho_s}$ for the distributions that are induced by pinching them to the eigenbasis of $\sigma$. Let a sequence $(P_n)_{n\in\nn}$ of projections defined by
\begin{align}
P_n:=\sum_{(\bar f,\bar\lambda)\in\Lambda_{\eps}}P^n_{f,\lambda},
\end{align}
where the frequency typical subspaces are defined with respect to $\{\tilde e_i\}_{i=1}^d$ and
\begin{align}
\Lambda_{\eps}&:=\{(p,q)\in\mathfrak P([d]\times[d])\ \left|\right.\ \exists\ s\in\bS\ :\ \|p-\tilde r_s\|\leq\eps\ \mathrm{and}\ \|q-r_s\|\leq\eps\}
\end{align}
and the $\tilde r_s$ are the probability distributions defined by $\tilde r_s(i):=\langle\tilde e_i,\rho_s\tilde e_i\rangle$. For this sequence, both
\begin{align}\label{eqn-1}
\inf_{s\in\bS}\tr\{P_n\rho_s^{\otimes n}\}\geq1-2^{-n(\alpha\eps^2-\frac{d^2}{n}\log(2n))}
\end{align}
and
\begin{align}\label{eqn-2}
\frac{1}{n}\log\tr\{P_n\sigma^{\otimes n}\}\leq-\inf_{s\in\bS}D(\rho_s\|\sigma)+\Theta(n,\eps,d,\sigma)
\end{align}
are true. In addition, for every $0\leq A\leq\eins^{\otimes n}$ satisfying $\inf_{s\in\bS}\tr\{A_n\rho_s^{\otimes n}\}\geq1-\nu$ (for some $\nu\in(0,1)$) it holds that
\begin{align}
-\inf_{s\in\bS}D(\rho_s\|\sigma)-\Theta'(n,\nu,d)\leq\frac{1}{n}\log\tr\{A_n\sigma^{\otimes n}\}.
\end{align}
The functions $\Theta$ and $\Theta'$ are defined by
\begin{align}
&\Theta(n,\eps,d,\sigma):=\frac{d^2}{n}\log(2n)+\eps|\log(\eps/d)|+d\eps\cdot\max_i|\log t_i|,\\
&\Theta'(n,\nu,d):=\Theta(n,n^{-\frac{1}{4}},d)-\frac{2d^6}{n}\log(2n)+\frac{1}{n}\log(\frac{1-\nu-2^{-\alpha\sqrt{n}}}{(2n)^{2d^2}}).
\end{align}
This proves the quantum Stein's Lemma by replacing $\eps$ with a suitably chosen sequence $(\eps_n)_{n\in\nn}$ that satisfies $\eps_n=\frac{1}{\sqrt{\alpha}}\sqrt{\frac{1}{n}\log(\nu)+\frac{d^2}{n}\log(2n)}$ for all $n\in\nn$.
\end{theorem}
\begin{remark}
By linearity of the trace, this implies that we can distinguish the hypotheses $\conv(\{\rho_s^{\otimes n}\}_{s\in\bS})$ and $\sigma^{\otimes n}$ at the same rate.
\end{remark}
Using Theorem \ref{theorem:main result}, we can prove the following result that relates the Kostka numbers to quantum states:
\begin{theorem}\label{theorem:kostka-numbers}
First, if $K_{f,\lambda}>0$ then there exists a quantum state with spectrum $\bar\lambda$ and pinching $\bar f$.\\
Second, if $\rho$ is a quantum state with spectrum $r_\rho$ and pinching $\tilde r_\rho$, then for every $\eps>0$ there exists $N\in\nn$ and for all $n\geq N$ there is $\lambda\in\mathbb Y_{d,n}$ and a corresponding frequency such that $\|\bar\lambda-r_\rho\|\leq n^{-1/4}$ and $\|\bar f-\tilde r_\rho\|\leq n^{-1/4}$ and $K_{f,\lambda}>0$.
\end{theorem}
Our next result concerns the arbitrarily varying quantum source. The proof of this theorem rests on the proof of Theorem \ref{theorem:main result}, so that the functions and objects defined in that theorem are used here as well:
\begin{theorem}\label{theorem:main-result-II}
Under the same preliminaries as above we have: For a finite set $\mathfrak S'=\{\rho_s\}_{s\in\bS}$ and with the projections $P_n$ defined as in Theorem \ref{theorem:main result} using the set $\mathfrak S:=\conv(\mathfrak S')$, the estimates
\begin{align}
\tr\{P_n\rho_{s^n}\}\geq1-2^{-n\alpha\eps^2+\frac{2d^2+|\bS|}{n}\log(2n)}
\end{align}
and
\begin{align}
\frac{1}{n}\log\tr\{P_n\sigma^{\otimes n}\}\leq-\min_{\rho\in\conv(\mathfrak S')}D(\rho\|\sigma)+\Theta(n,\eps,d,\sigma)+\frac{|\bS|}{n}\log(2n)
\end{align}
and for every operator $0\leq A_n\leq\eins^{\otimes n}$ satisfying $\min_{s^n\in\bS^n}\tr\{A_n\rho_{s^n}\}\geq1-\nu$ for some $\nu\in(0,1)$ we have
\begin{align}
-\min_{\rho\in\conv(\mathfrak S')}D(\rho\|\sigma)-\Gamma(n,\nu,d,\sigma)\leq\frac{1}{n}\log\tr\{A_n\sigma^{\otimes n}\}
\end{align}
where
\begin{align}
&\Gamma(n,\nu,d,\sigma):=\Theta(n,n^{-1/4},d,\sigma)+\Gamma'(n,\nu,d)+\frac{8d^6}{n}\log(2n)\\
&\Gamma'(n,\nu,d):=\frac{1}{n}\log[(1-\nu-2^{-n[\alpha(n^{-1/4}-2|\bS|n^{-1})-\frac{|\bS|-2d^2}{n}\log(2n)]})(2n)^{-8d^2}].
\end{align}
This proves Theorem \ref{thm:mainresult-short-form} by replacing $\eps$ with $\eps_n:=\frac{1}{\sqrt{\alpha}}\sqrt{\frac{1}{n}\log(\nu)+\frac{d^2+|\bS|}{n}\log(2n)}$.\\
In case that $|\bS|=\infty$, the following is true: For the sequence $(\delta_n)_{n\in\nn}$ defined by $\delta_n:=12n^{-1/4d^2}$ ($n\in\nn$) and the depolarizing channel $\cn_{\delta_n}:\mathcal B(\mathbb C^d)\to\mathcal B(\mathbb C^d)$, $a\mapsto(1-\delta_n)a+\delta_n\frac{1}{d}\eins$ we have
\begin{align}
\inf_{s^n\in\bS^n}\tr\{\cn_{\delta_n}^{\otimes n}(P_n)\rho_{s^n}\}\geq1-2^{-n(\alpha\eps^2-\frac{d^2+\sqrt{n}}{n}\log(2n))}
\end{align}
and
\begin{align}
\tr\{\cn_{\delta_n}^{\otimes n}(P_n)\sigma^{\otimes n}\}\leq2^{-n(\inf_{\rho\in\conv(\mathfrak S')}D(\rho\|\sigma)+d\cdot\log(1-\delta_n)-\Theta(n,\eps,d,\sigma))}
\end{align}
as well as, for every $\nu\in[0,1)$ and $0\leq A_n\leq\eins^{\otimes n}$: If $\inf_{s^n\in\bS^n}\tr\{A_n\rho_{s^n}\}\geq1-\nu$, then
\begin{align}
\tr\{A_n\sigma^{\otimes n}\}\leq2^{-n(\inf_{\rho\in\conv(\mathfrak S')}D(\rho\|\sigma)-\Gamma(n,\nu,d,\cn_{\delta_n}(\sigma),\sqrt{n})-c(\delta_n))}.
\end{align}
where $\tau\log(d)+h(\tau)+\tau\log(1/t_d)=:c(\tau,d,\sigma)$. Note that, for an arbitrary $\nu\in(0,1)$, we have $\lim_{n\to\infty}\Gamma(n,\nu,d,\cn_{\delta_n}(\sigma),\sqrt{n})=0$ and that also $\lim_{n\to\infty}c(\delta_n,d,\sigma)=0$.
This proves Theorem \ref{thm:mainresult-short-form} by replacing $\eps$ with $\eps_n=\sqrt{1}{\alpha}\sqrt{\frac{1}{n}\log(\nu)+\frac{d^2+\sqrt{n}}{n}\log(2n)}$.
\end{theorem}
\begin{remark}
By linearity of the trace, this implies that we can distinguish the hypotheses $\conv(\{\rho_{s^n}\}_{s^n\in\bS^n})$ and $\sigma^{\otimes n}$ at the same rate.
\end{remark}
The proof of Theorem \ref{theorem:main-result-II} reveals an interesting Lemma, which is a generalization of the well-known concentration theorem for spectral estimation \cite{thm2}. We state it here exclusively for sake of completeness:
\begin{lemma}\label{lemma:state-estimation} Let $\rho_{s^n}:=\otimes_{i=1}^n\rho_{s_i}$ for a finite set $\{\rho_s\}_{s\in\bS}$ of states. Let $\bar\rho:=\sum_s\frac{1}{n}N(s|s^n)\rho_s$ and let $\lambda\in\mathbb Y_{d,n}$. Then
\begin{align}
\tr\{P_\lambda\rho_{s^n}\}\leq(2n)^{|\bS|+d^2}2^{-n(D(\bar\lambda\|\mathrm{spec}(\bar\rho))}.
\end{align}
\end{lemma}
Our last result is of negative character:
\begin{theorem}\label{theorem:additional-result}
Let $n\in\nn$, $\eps\geq0$ and a  permutation invariant operator $A\in\mathcal B((\mathbb C^d)^{\otimes n})$ which is bounded by $0\leq A_n\leq\eins^{\otimes n}$ be given. Then
\begin{align}
\min_{\rho\in\cs(\mathbb C^d)}\tr\{A_n\rho^{\otimes n}\}\geq1-\eps\qquad\Rightarrow\qquad A_n\geq(1-\eps\cdot(2dn)^{4d^2})\eins^{\otimes n}
\end{align}
and, by applying this result to $(\eins^{\otimes n}-A)$ instead of $A$:
\begin{align}
\max_{\rho\in\cs(\mathbb C^d)}\tr\{A\rho^{\otimes n}\}\leq\eps\qquad\Rightarrow\qquad A_n\leq\eps\cdot(2dn)^{4d^2}\eins^{\otimes n}.
\end{align}
\end{theorem}
\begin{remark}
A look at the proof of this theorem reveals that the decomposition $V_\lambda=\oplus_f V_{f,\lambda}$ is still not sufficiently fine for proving the theorem, indicating that there is some more relevant information to be gained by using POVMs that employ projections onto irreducible subspaces rather than the (generally reducible) ones that we employ. An even stronger method of types is thus still a possibility.\\
An operational implication of Theorem \ref{theorem:additional-result} result is that any such sequence $((A_n,\eins^{\otimes n}-A_n))_{n\in\nn}$ of POVMs is useless for distinguishing between separable states and entangled states on $(\mathbb C^d)^{\otimes n}$, for large values of $n$: In the light of our Theorem \ref{theorem:main-result-II}, and since $\cs(\mathbb C^d)=\conv(\cs(\mathbb C^d))$, we see that any measurement scheme that successfully distinguishes between the hypotheses ``some fixed entangled state'' versus ``any separable state'' has probability of successfully detecting the second hypothesis converging to one only polynomially fast, or it cannot be permutation invariant.
\end{remark}
We arrive at our two main results by exploiting the simple observation that every frequency typical subspace $V_f$ carries a representation of the symmetric group $S_n$ and can therefore be decomposed into irreducible subspaces of $S_n$.
\end{section}
%%%%%%%%%%%%%%%%%%%%%%%%%%%%%%%%%%%%%%%%%%%%%%%%%%%%%%%%%%%%%%%%%%%%%%%%%%%%%%%%%%%%%%%%
%%%%%%%%%%%%%%%%%%%%%%%%%%%%%%%%%%%%%%%%%%%%%%%%%%%%%%%%%%%%%%%%%%%%%%%%%%%%%%%%%%%%%%%%
\begin{section}{\label{sec:Proofs}Proofs}
\begin{proof}[Proof of Theorem \ref{theorem:main result}]
This proof is split into several subsections. We start with the error exponent.\\
\emph{Proof that we have the correct error exponent in Theorem \ref{theorem:main result}.} For a projection $P_{f,\lambda}$ we have - including the case that $t_{d-k},\ldots,t_d=0$ for some $k>0$ through the convention $2^{-\infty}=0$ and the case that $\dim(F_\lambda)>0$ but $\tr\{P_{f,\lambda}\}=0$ that can occur if $V_f$ does not carry any of the representations corresponding to $\lambda$ -
\begin{align}\label{eqn:start-of-error-exponent}
\tr\{P_{f,\lambda}\sigma^{\otimes n}\}&\leq m_{\lambda,n}\dim(F_{\lambda})2^{n\sum_{i=1}^d\frac{1}{n}f_i\log(t_i)}\\
&\leq m_{\lambda,n}2^{n(H(\overline \lambda)+\sum_{i=1}^d\frac{1}{n}f_i\log t_i)}.
\end{align}
If, for some $s\in\bS$, both $\|\bar\lambda-r_s\|\leq\eps$ and $\|f-\tilde r_s\|\leq\eps$, then equation (\ref{eqn:start-of-error-exponent}) leads to
\begin{align}
\tr\{P_{f,\lambda}\sigma^{\otimes n}\}&\leq m_{\lambda,n}2^{n(H(r_s)+\eps|\log(\eps/d)|+\sum_{i=1}^d\tilde r_s(i)\log t_i+d\eps\cdot\max_i|\log t_i|)}.
\end{align}
But for an orthonormal eigenbasis $\{e_i^{s}\}_{i=1}^d$ of $\rho_s$ we have $\tilde r_s(i)=\sum_j|\langle \tilde e_i,e_j^s\rangle|^2r_s(j)$ and so
\begin{align}
H(r)+\sum_{i=1}^d\tilde r_s(i)\log t_i=-D(\rho_s\|\sigma),
\end{align}
and for the projection $P_n$ this means that
\begin{align}
\tr\{P_n\sigma^{\otimes n}\}&=\sum_{(f,\lambda)\in\Lambda_\eps}\tr\{P_{f,\lambda}\sigma^{\otimes n}\}\\
&\leq(2n)^{d^2}\max_{(f,\lambda)\in\Lambda_\eps}\tr\{P_{f,\lambda}\sigma^{\otimes n}\}\\
&\leq(2n)^{d^2}\sup_{s\in\bS}2^{-n(D(\rho_s\|\sigma)-\eps|\log(\eps/d)|-d\eps\cdot\max_i|\log t_i|)}\\
&=2^{-n(\inf_{s\in\bS}D(\rho_s\|\sigma)-\frac{d^2}{n}\log(2n)-\eps|\log(\eps/d)|-d\eps\cdot\max_i|\log t_i|)}\\
&=2^{-n(\inf_{s\in\bS}D(\rho_s\|\sigma)-\Theta(n,\eps,d,\sigma))}\label{eqn:end-of-error-exponent},
\end{align}
which is what we wanted to show.
\\\\
%%%%%%%%%%%%%%%%%%%%%%%%%%%%%%%%%%%%%%%%%%%%%%%%%%%%%%%%%%%%%%%%%%%%%%%%%%%%%%%%
%%%%%%%%%%%%%%%%%%%%%%%%%%%%%%%%%%%%%%%%%%%%%%%%%%%%%%%%%%%%%%%%%%%%%%%%%%%%%%%%
\emph{Proof that we can identify every state in our given set with high probability.} For every $s\in\bS$, define the projections $P_{n,s}:=\sum_{f:\|\bar f-\tilde r_s\|\leq\eps}\sum_{\lambda:\|\bar\lambda-r_s\|\leq\eps} P_{f,\lambda}$. These satisfy the lower bound
\begin{align}
\forall\ s\in\bS,\qquad P_n\geq P_{n,s}.
\end{align}
Another inequality we will need is an upper bound that is valid for all $P_{f,\lambda}$ satisfying $\tr\{P_{f,\lambda}P_{n,s}\}=0$. It generally holds that $P_{f,\lambda}\leq \sum_\lambda P_{f,\lambda}=:P_f$ and $P_{f,\lambda}\leq\sum_f P_{f,\lambda}=:P_\lambda$ and thus by \cite{thm2} and the basic estimate Lemma 2.6 in \cite{csiszar-koerner} followd by Pinsker's inequality we get the upper bound
\begin{align}
\tr\{\rho_s^{\otimes n}P_{f,\lambda}\}&\leq\min\{\tr\{\rho_s^{\otimes n}P_f\},\tr\{\rho_s^{\otimes n}P_\lambda\}\}\\
&\leq(2n)^{d^2}\min\{2^{-nD(\bar\lambda\|r_s)},2^{-nD(\bar f\|\tilde r_s)}\}\\
&\leq (2n)^{d^2}2^{-n\alpha\eps^2}.
\end{align}
Set $X_s:=\{(f,\lambda):\|\bar f-\tilde r_s\|>\eps\ \bigwedge\ \|\bar\lambda-r_s\|>\eps\}$. From definition of our basic object it can be seen that the inequality $\sum_f\sum_\lambda P_{f,\lambda}=\eins^{\otimes n}$ holds, so for an arbitrary $s\in\bS$ we have
\begin{align}
1&=\tr\{\rho_s^{\otimes n}\sum_f\sum_\lambda P_{f,\lambda}\}\\
&=\tr\{\rho_s^{\otimes n}P_{n,s}\}+\tr\{\rho_s^{\otimes n}(\eins_{\mathbb C^d}^{\otimes n}-P_{n,s})\}\\
&\leq\tr\{\rho_s^{\otimes n}P_{n,s}\}+(2n)^{d^2}\max_{(f,\lambda)\in X_s}\tr\{\rho_s^{\otimes n}P_{f,\lambda}\}\\
&\leq\tr\{\rho_s^{\otimes n}P_{n,s}\}+(2n)^{2d^2}2^{-n\alpha\eps^2}\label{eqn:identify-even-more}\\
&\leq\tr\{\rho_s^{\otimes n}P_{n}\}+(2n)^{2d^2}2^{-n\alpha\eps^2}.
\end{align}
It follows
\begin{align}\label{eqn-for-kostka-existence}
\inf_{s\in\bS}\tr\{\rho_s^{\otimes n}P_n\}\geq1-2^{-n(\alpha\eps^2-\frac{2d^2}{n}\log(2n))}.
\end{align}
Therefore,
\begin{align}
\lim_{n\to\infty}\inf_{s\in\bS}\tr\{P_n\rho_s^{\otimes n}\}=1,
\end{align}
as required. By linearity of the trace, the results immediately imply that we can also (asymptotically) distinguish the sequence $(\conv(\{\rho^{\otimes n}_s\}_{s\in\bS})_{n\in\nn}$ of sets from the sequence $(\sigma^{\otimes n})_{n\in\nn}$.
\\\\
%%%%%%%%%%%%%%%%%%%%%%%%%%%%%%%%%%%%%%%%%%%%%%%%%%%%%%%%%%%%%%%%%%%%%%%%%%%%%%%%%%%%%%%%
%%%%%%%%%%%%%%%%%%%%%%%%%%%%%%%%%%%%%%%%%%%%%%%%%%%%%%%%%%%%%%%%%%%%%%%%%%%%%%%%%%%%%%%%
\emph{Proof of the lower bound on the error exponent in Theorem \ref{theorem:main result}.} Let an operator $0\leq A\leq\eins^{\otimes n}$ be given that satisfies
\begin{align}
\inf_{s\in\bS}\tr\{A\rho_s^{\otimes n}\}\geq1-\nu
\end{align}
for some $\nu\in(0,1)$. Note that we can without loss of generality assume that $A$ is permutation-invariant, since permuting an $A$ that is not invariant does change neither $\tr\{A\rho^{\otimes n}\}$ nor $\tr\{A\sigma^{\otimes n}\}$. But then by Schur-Weyl duality we can write
\begin{align}
A=\sum_\lambda A_\lambda,
\end{align}
where $0\leq A_\lambda\leq c_\lambda P_\lambda$ for some set of real numbers $c_\lambda$. This leads us to
\begin{align}
\tr\{\sigma^{\otimes n}A_\lambda\}&=\sum_{\lambda,f,j}2^{n\sum_i\bar f(i)\log t_i}\tr\{P_{f,\lambda,j}A_\lambda P_{f,\lambda,j}\},
\end{align}
where $P_{f,\lambda,j}$ are any set of mutually \emph{orthogonal} and \emph{irreducible} projections such that $P_{f,\lambda}=\sum_jP_{f,\lambda,j}$ for all $(f,\lambda)$. On the other hand by Lemma 5 in \cite{hayashi01} (which is identical to Lemma 9 in \cite{hayashi02}) we get
\begin{align}
\sum_{f,\lambda,j}P_{f,\lambda,j}A P_{f,\lambda,j}\geq(2n)^{-d^2}A,
\end{align}
and it follows
\begin{align}
1-\nu\leq(2n)^{d^2}\sum_{f,\lambda,j}\tr\{P_{f,\lambda,j}A P_{f,\lambda,j}\rho^{\otimes n}\}.
\end{align}
It is clear that $P_{f,\lambda,j}A P_{f,\lambda,j}=c_{f,\lambda,j}P_{f,\lambda,j}$ for every triple $(f,\lambda,j)$ and a suitable choice of coefficients $0\leq c_{f,\lambda,j}\leq1$ by permutation-invariance of $A$. Therefore there is at least one $(f,\lambda,j)$ such that $\|\bar f-\tilde r\|\leq n^{-1/4}$, $\|\bar\lambda-r\|\leq n^{-1/4}$ and
\begin{align}
c_{f,\lambda,j}\geq(2n)^{-2d^2}(1-\nu-2^{-\alpha\sqrt{n}}).
\end{align}
But this implies that for this triple $(f,\lambda,j)$ we have (and at this point one sees again that it is crucial that we defined the $V_f$ using the eigenbasis of $\sigma$, since this ensures that $\sigma^{\otimes n}$ and every one of the $P_{f,\lambda,j}$ commute)
\begin{align}
\tr\{\sigma^{\otimes n}A\}&\geq\tr\{\sigma^{\otimes n}A_\lambda\}\\
&=\sum_{f',i}\tr\{\sigma^{\otimes n}P_{f',i}A_\lambda P_{f',i}\}\\
&\geq2^{n\sum_i\bar f(i)\log t_i}\tr\{P_{f,\lambda,j}A_\lambda P_{f,\lambda,j}\}\\
&=2^{n\sum_i\bar f(i)\log t_i}\tr\{P_{f,\lambda,j}\}c_{f,\lambda,j}\\
&\geq2^{n\sum_i\bar f(i)\log t_i}\dim(F_\lambda)\frac{1-\nu-2^{-\alpha\sqrt{n}}}{(2n)^{2d^2}}\\
&\geq2^{n(\sum_i\bar f(i)\log(t_i)-H(\bar\lambda)-\frac{2d^6}{n}\log(2n))}\frac{1-\nu-2^{-\alpha\sqrt{n}}}{(2n)^{2d^2}}\\
&\geq2^{-n(D(\rho\|\sigma)-\Theta(n,n^{-\frac{1}{4}},d)-\frac{2d^6}{n}\log(2n)+\frac{1}{n}\log(\frac{1-\nu-2^{-\alpha\sqrt{n}}}{(2n)^{2d^2}}))},
\end{align}
So the function
\begin{align}
\Theta'(n,\nu,d):=\Theta(n,n^{-\frac{1}{4}},d)-\frac{2d^6}{n}\log(2n)+\frac{1}{n}\log(\frac{1-\nu-2^{-\alpha\sqrt{n}}}{(2n)^{2d^2}})
\end{align}
does the job.
\end{proof}
%%%%%%%%%%%%%%%%%%%%%%%%%%%%%%%%%%%%%%%%%%%%%%%%%%%%%%%%%%%%%%%%%%%%%%%%%%%%%%%%%%%%%%%%
%%%%%%%%%%%%%%%%%%%%%%%%%%%%%%%%%%%%%%%%%%%%%%%%%%%%%%%%%%%%%%%%%%%%%%%%%%%%%%%%%%%%%%%%
\begin{proof}[Proof of Theorem \ref{theorem:kostka-numbers}]
Let us first state some known facts about Kostka numbers. It holds $K_{f,\lambda}>0\ \Leftrightarrow\ f\prec\lambda$. But $K_{f,\lambda}$ is just the number of ways to fill the tableau $\lambda$ with numbers from $[d]$ such that the numbers in the filling are weakly increasing in each row, read from left to right, and strictly increasing along the columns (read from top to bottom). Thus $K_{f,\lambda}>0\ \Leftrightarrow\ \tr\{P_{f,\lambda}\}>0$. There is actually another important equivalence that we shall use: $f\prec\lambda\ \Leftrightarrow\ f=D\lambda$ for some doubly stochastic matrix $D$.\\
Now assume $K_{f,\lambda}>0$ holds. Thus $f=D\lambda$ for some doubly stochastic matrix $D$. In this case we may define a state $\rho$ by the following procedure: Pick any basis $\{e_i\}_{i=1}^d$. Choose a unitary matrix $U=\sum_{i,j}u_{ij}|e_i\rangle\langle e_j|$ such that for the matrix coefficients of $D$ it holds $D_{ij}=|u_{ij}|^2$ for all $i,j\in[d]$. This can be done, thanks to Horn's Lemma \cite{horn-lemma}. Define $\rho:=\sum_i\bar\lambda_i|e_i\rangle\langle e_i|$, then $\langle e_i,U\rho U^\dag e_i\rangle=f(i)$ for all $i\in[d]$, and this means that the pinching of $\rho$ to the basis $\{U^\dag e_i\}_{i=1}^d$ satisfies our requirements.\\
Now take any state $\rho$ with spectrum $r_\rho$ and pinching $\tilde r_\rho$ in some (not exactly specified) basis. If $\rho=\frac{1}{d}\eins$ there is nothing to prove, since every pinching then satisfies $\tilde r_\rho=r_\rho$ and obviously choosing $D=Id$ is consistent with elementary fact that $K_{nr_\rho,nr_\rho}>0$ for every $n$ such that $n/d\in\nn$.\\
Thus, $\rho\neq\frac{1}{d}\eins$. Choosing $\sigma=\frac{1}{d}\eins$ it follows that $D(\rho\|\sigma)>0$ and using our projections $P_n$ we see from equation (\ref{eqn-for-kostka-existence}) that there is $N\in\nn$ such that for all $n\geq N$ there is a projection $P_{f,\lambda}\neq0$ that has the property $\|\bar f-\tilde r_\rho\|\leq n^{-\frac{1}{4}}$ and $\|\bar\lambda-r_\rho\|\leq n^{-\frac{1}{4}}$, proving our claim.
\end{proof}
%%%%%%%%%%%%%%%%%%%%%%%%%%%%%%%%%%%%%%%%%%%%%%%%%%%%%%%%%%%%%%%%%%%%%%%%%%%%%%%%%%%%%%%%
%%%%%%%%%%%%%%%%%%%%%%%%%%%%%%%%%%%%%%%%%%%%%%%%%%%%%%%%%%%%%%%%%%%%%%%%%%%%%%%%%%%%%%%%
\begin{proof}[Proof of Theorem \ref{theorem:main-result-II}] We will first show the statements in the case that $\mathfrak S'$ is finite, and then apply them to the case of an infinite set $\mathfrak S'$. The issues arising in the latter setting are mainly due to 'roundness' of the set of states. It could possibly happen that $\mathfrak S'$ contains parts of the boundary of the set of $\cs(\mathbb C^d)$, and it can then (thinking of the Bloch sphere representation is helpful here) be impossible to include $\mathfrak S'$ into the convex hull of some finite set that lies entirely within the set of states itself.\\
\emph{Robustification, Part I - Finite $\bS$.} Let a state $\rho_{\hat s^n}$ be given, and let $\hat s^n$ have type $N:\bS\to\nn$ (meaning that the short-hand $N(\cdot)=N(\cdot|\hat s^n)$ is valid). Let $p\in\mathfrak P(\bS)$ be defined by $p(s):=\frac{1}{n}N(s)$. It then holds (see \cite{csiszar-koerner}, the proof of Lemma 2.3) that $p^{\otimes n}(T_N)\geq (2n)^{-|\bS|}$. Since our measurement is permutation-invariant, it follows
\begin{align}
\tr\{(\eins-P_n)(\sum_{s\in\bS}p(s)\rho_s)^{\otimes n}\}&=\tr\{(\eins-P_n)\sum_{s^n\in\bS^n}p^{\otimes n}(s^n)\rho_{s^n}\}\\
&\geq\tr\{(\eins-P_n)\sum_{s^n\in T_N}p^{\otimes n}(s^n)\rho_{s^n}\}\\
&=\tr\{(\eins-P_n)\sum_{s^n\in T_N}p^{\otimes n}(s^n)\rho_{\hat s^n}\}\\
&=p^{\otimes n}(T_N)\tr\{(\eins-P_n)\rho_{\hat s^n}\}\\
&\geq (2n)^{-|\bS|}\tr\{(\eins-P_n)\rho_{\hat s^n}\},
\end{align}
implying, via equation (\ref{eqn-for-kostka-existence}) and our choice of $P_n$, that
\begin{align}
\tr\{P_n\rho_{\hat s^n}\}&=1-\tr\{(\eins-P_n)\rho_{\hat s^n}\}\\
&\geq1-(2n)^{|\bS|}\tr\{(\eins-P_n)(\sum_{s\in\bS}p(s)\rho_s)^{\otimes n}\}\\
&\geq1-(2n)^{|\bS|}2^{-n(\alpha\eps^2+\frac{2d^2}{n}\log(2n))}\\
&\geq1-2^{-n(\alpha\eps^2+\frac{2d^2+|\bS|}{n}\log(2n))}.
\end{align}
By linearity of the trace, this immediately implies that we can distinguish the sequence $(\conv(\{\rho_{s^n}\}_{s^n\in\bS^n}\}))_{n\in\nn}$ of sets from the sequence $(\sigma^{\otimes n})_{n\in\nn}$.\\
We still have to show that we have the right error exponent, but this follows immediately from the Sanov case by noting that for arbitrary $s^n\in\bS^n$ and the type $N(\cdot):=N(\cdot|s^n)$ we have by Lemma 2.3 in \cite{csiszar-koerner}
\begin{align}
\tr\{P_n\rho_{s^n}\}\leq(2n)^{|\bS|}\tr\{P_n(\sum_{s\in\bS}\bar N(s)\rho_s)^{\otimes n}\}.
\end{align}
\emph{Proof of the converse part for finite $|\bS|$.} Let $\tr\{A\rho_{s^n}\}\geq1-\nu$ for all $s^n\in\bS^n$. Like before, we can without loss of generality assume that $A$ is permutation-invariant. Then again by Lemma 5 in \cite{hayashi02} we get
\begin{align}
A&\leq(2n)^{d^2}\sum_{f,\lambda,j}c_{f,\lambda,j}P_{f,\lambda,j},
\end{align}
where $P_f:=\sum_\lambda P_{f,\lambda}$, $P_\lambda:=\sum_fP_{f,\lambda}$ and $P_fP_\lambda=P_\lambda P_f$ holds and the index $j$ refers to \emph{any} decomposition $V_{f,\lambda}=\oplus_{\lambda,j}V_{f,\lambda,j}$ into \emph{irreducible} components. Obviously, $P_{f,\lambda,j}$ then denote the orthogonal projections onto the $V_{f,\lambda,j}$.\\
Let us pick $p\in\mathfrak P(\bS)$ and an $s^n\in\bS^n$ having type abbreviated by $N(\cdot)=N(\cdot|s^n)$. We may additionally assume that $N$ satisfies $\|\bar N-p\|\leq \delta_n$ (where $\delta_n=2|\bS|/n$). Then assume that a $\lambda\in\mathbb Y_{d,n}$ satisfies $\|\bar\lambda-\mathrm{spec}(\bar\rho)\|>\delta>0$ for some $\delta$ and the state $\bar\rho:=\sum_sp(s)\rho_s$. For such $\lambda$ it follows from \cite{csiszar-koerner}, Lemma 2.3, that
\begin{align}
\tr\{P_\lambda\rho_{s^n}\}&=\frac{1}{|T_N|}\sum_{s^n\in T_N}\tr\{P_\lambda \rho_{s^n}\}\\
&\leq(2n)^{|\bS|}\bar N^{\otimes n}(T_N)\sum_{s^n\in T_N}\tr\{P_\lambda\rho_{s^n}\}\\
&\leq(2n)^{|\bS|}\sum_{s^n\in\bS^n}\bar N^{\otimes n}(s^n)\tr\{P_\lambda\rho_{s^n}\}\\
&=(2n)^{\bS|}\tr\{P_\lambda(\sum_s\bar N(s)\rho_s)^{\otimes n}\}\\
&\leq(2n)^{|\bS|}(2n)^{d^2}2^{-nD(\bar\lambda\|\mathrm{spec}(\sum_s\bar N(s)\rho_s))}\\
&\leq(2n)^{|\bS|+d^2}2^{-n\alpha\|\bar\lambda-\mathrm{spec}(\sum_s\bar N(s)\rho_s)\|^2}\\
&\leq(2n)^{|\bS|+d^2}2^{-n\alpha(\|\bar\lambda-\mathrm{spec}(\bar\rho)\|-\|\mathrm{spec}(\bar\rho)-\mathrm{spec}(\sum_s\bar N(s)\rho_s)\|)^2}\\
&\leq(2n)^{|\bS|+d^2}2^{-n\alpha(\|\bar\lambda-\mathrm{spec}(\bar\rho)\|-\delta_n)^2}\\
&\leq(2n)^{|\bS|+d^2}2^{-n\alpha(\|\bar\lambda-\mathrm{spec}(\bar\rho)\|^2-2\delta_n)}\\
&\leq(2n)^{|\bS|+d^2}2^{-n\alpha(\delta^2-2\delta_n)},
\end{align}
where the third inequality counted from below is due to inequality (1) in \cite{bhatia-spectral-variation}. In the same way, we get that for all frequencies $f$ satisfying $\|\bar f-\tilde r_{\bar\rho}\|>\delta$ we have the following inequality:
\begin{align}
\tr\{P_f\rho_{s^n}\}&\leq(2n)^{|\bS|+d^2}2^{-n\alpha(\delta^2-2\delta_n)}
\end{align}
We use these results to get (using the abbreviation $\bar\rho:=\sum_sp(s)\rho_s$) the following:
\begin{align}
1-\nu&\leq(2n)^{d^2}\sum_{f,\lambda,j}c_{f,\lambda,j}\tr\{P_{f,\lambda,j}\rho_{s^n}\}\\
&\leq(2n)^{d^2}\sum_{f,\lambda,j}c_{f,\lambda,j}\tr\{P_{f,\lambda,j}\rho_{s^n}\}\\
&\leq(2n)^{d^2}\sum_{\|\bar f-\tilde r_{\bar\rho}\|\leq\delta}\sum_{\|\bar\lambda-r_{\bar\rho}\|\leq\delta}c_{f,\lambda,j}\tr\{P_{f,\lambda,j}\rho_{s^n}\}+(2n)^{|\bS|+2d^2}2^{-n\alpha(\delta^2-2\delta_n)}\\
&\leq (2n)^{2d^2}\max\{c_{f,\lambda,j}:\|\bar\lambda-r_{\bar\rho}\|\leq\delta,\ \|\bar f-\tilde r_{\bar\rho}\|\leq\delta\}+2^{-n[\alpha(\delta^2-2\delta_n)-\frac{|\bS|+2d^2}{n}\log(2n)]},
\end{align}
and choosing $\delta=n^{-1/8}$ we get (as before) that there exists a triple $(f,\lambda,j)$ such that $\|\bar\lambda-r_{\bar\rho}\|\leq n^{-1/8}$ and $\|\bar f-\tilde r_{\bar\rho}\|\leq n^{-1/8}$ and at the same time
\begin{align}
(1-\nu-2^{-n[\alpha(n^{-1/4}-4|\bS|n^{-1})-\frac{|\bS|-2d^2}{n}\log(2n)]})(2n)^{-2d^2}\leq c_{f,\lambda,j}.
\end{align}
Let us abbreviate this by writing instead that, for the quantity
\begin{align}
\Gamma'(n,d,\nu,|\bS|):=\frac{1}{n}\log[(1-\nu-2^{-n[\alpha(n^{-1/4}-4|\bS|n^{-1})-\frac{|\bS|-2d^2}{n}\log(2n)]})(2n)^{-2d^2}],
\end{align}
we have
\begin{align}
c_{f,\lambda,j}\geq2^{n\Gamma'(n,d,\nu,|\bS|)}.
\end{align}
It is also important to note that $\lim_{n\to\infty}\Gamma'(n,d,\nu,|\bS|)=0$, if the other values remain fixed. This leads to
\begin{align}
\tr\{A\sigma^{\otimes n}\}&\geq c_{f,\lambda,j}\tr\{P_{f,\lambda,j}\sigma^{\otimes n}\}\\
&\geq2^{n\Gamma'(n,d,|\bS|)}2^{n(H(\bar\lambda)-\frac{2d^6}{n}\log(2n))}2^{n\sum_i\bar f(i)\log t_i}\\
&\geq2^{-n(D(\bar\rho\|\sigma)+\Theta(n,n^{-\frac{1}{8}},d,\sigma)-\Gamma'(n,d,\nu,|\bS|)+\frac{2d^6}{n}\log(2n))}\\
&=2^{-nD(\bar\rho\|\sigma)+\Gamma(n,\nu,d,\sigma,|\bS|))},
\end{align}
with the obvious definition of the function $\Gamma$. Note that $\lim_{n\to\infty}\Gamma(n,\nu,d,\sigma,|\bS|)=0$ for all $\nu\in(0,1)$. But above estimate holds for all $\bar\rho\in\conv(\mathfrak S')$, so
\begin{align}
\tr\{A\sigma^{\otimes n}\}&\geq2^{-n(\inf_{\rho\in\conv(\mathfrak S')}D(\rho\|\sigma)-\Gamma(n,\nu,d,\sigma,|\bS|))}.
\end{align}
%%%%%%%%%%%%%%%%%%%%%%%%%%%%%%%%%%%%%%%%%%%%%%%%%%%%%%%%%%%%%%%%%%%%%%%%%%%%%%%%%%%%%%%%
%%%%%%%%%%%%%%%%%%%%%%%%%%%%%%%%%%%%%%%%%%%%%%%%%%%%%%%%%%%%%%%%%%%%%%%%%%%%%%%%%%%%%%%%
\ \\ \emph{Robustification, Part II - infinite $\bS$.} We will approximate the set $\bS$ by finite sets first, to which we then apply the previous result. We proceed as follows. First, we drag the set $\mathfrak S'$ a bit into the direction of the center of $\cs(\mathbb C^d)$ by applying an appropriate cptp map. The result will be that the so modified set $\mathfrak S''$ has positive set distance from the boundary of $\cs(\mathbb C^d)$. We can then find a finite set $\mathfrak S'''$ of states such that $\mathfrak S''\subset\conv(\mathfrak S''')$ holds. The adjoint of our cptp map applied to our invariant projections then defines our measurement. Unfortunately this means that our measurement operators in this scenario are no longer extremal.\\
Let us get started.\\
Without diving deeper into issues concerning distances between sets at this point it is clear that (for the \emph{convex} set $\mathfrak S:=\conv(\mathfrak S')$) the statement $\sigma\notin\mathfrak S$ implies that there is a small $\delta$ satisfying $1>\delta>0$ such that the channel
\begin{align}
\cn_\delta(\cdot):=(1-\delta)Id(\cdot)+\delta\frac{1}{d}\eins
\end{align}
has the property that $\sigma\notin\cn_\delta(\mathfrak S)$. Since the case $\sigma\in\mathfrak S$ is trivial, we may assume that $\sigma\notin\mathfrak S$ without loss of generality. Define $\mathfrak S'':=\cn_\delta(\mathfrak S)$. Now take any finite set $\mathfrak S'''=\{\hat\rho_s\}_{s\in\bS'''}$ such that $\mathfrak S''\subset\conv(\mathfrak S''')$. Then, applying the result from the previous proof to our new set $\conv(\mathfrak S''')$ we see that
\begin{align}\label{eqn:lower-bound-on-success-probability}
\tr\{P_n\rho_{s^n}'\}&\geq1-2^{-n(\alpha\eps^2+\frac{2d^2+|\bS'''|}{n}\log(2n))}
\end{align}
holds for all $\rho_{s^n}'$ for which $\rho_{s_i}'\in\conv(\mathfrak S''')$ for all $i=1,\ldots,n$. But for the transpose channel $\cn_\delta^\dag$ of $\cn_\delta$ this implies that for all $s^n\in\bS^n$ we have
\begin{align}
\tr\{\cn_\delta^{\dag\otimes n}(P_n)\rho_{s^n}\}&\geq1-2^{-n(\alpha\eps^2+\frac{d^2+|\bS'''|}{n}\log(2n))}.
\end{align}
Since $0\leq\cn_\delta^{\dag\otimes n}(P_n)\leq\eins^{\otimes n}$ holds we have succeeded in constructing a measurement that correctly identifies the AVQS $\mathfrak S'$. The next step naturally is to make $\delta$ very small, and hope that we then get closer and closer to the optimal exponent. This will of course come at a price: the lower bound (\ref{eqn:lower-bound-on-success-probability}) is dependent on the size $\bS'''$ of the number of extremal points in our finite approximation of $\mathfrak S'$ - and this number increases with $\delta$. We shall now quantify this.\\
Consider the supporting hyperplane of $\mathfrak S'\cup\{\frac{1}{d}\eins\}$ as a real vector space $W$ that contains all the sets $\mathfrak S,\mathfrak S',\mathfrak S''$. All operators in $a\in W$ obey $\tr\{a\}=1$. Obviously, this vector space does contain matrices that are not in $\cs(\mathbb C^d)$ - they have negative eigenvalues. We want to construct a finite approximation $\mathfrak S'''$ to $\mathfrak S''$ such that $\mathfrak S''\subset\conv(\mathfrak S''')\subset\cs(\mathbb C^d)$.\\
On $W$, we may still use the one-norm $\|\cdot\|$ as a norm. The intersection $B:=\rebd\cs(\mathbb C^d)\cap W$ satisfies $\mathfrak S''\subset\conv(B)$, and it is thus our goal to construct an approximation satisfying $\mathfrak S''\subset\conv(\mathfrak S''')\subset\conv(B)$. On $W$, a $\delta$-net in 1-norm of cardinality no more than $N(\delta)\leq(2/\delta)^{2d^2}$ can be found via the volume bound from \cite{milman-schechtman} (Lemma 2.6) in its application to quantum states (see e.g. \cite{bbn1}, Lemma 5.1 and note that quantum states are just special instances of quantum channels). We additionally need to show that the blow-up $(\mathfrak S'')_\eps$ is contained in $\conv(B)$ for small enough $\eps>0$. This is done by taking any normalized vector $x\in\mathbb C^d$ and computing, for some $a\in W$ satisfying $\|a-\cn_\delta(\rho_s)\|\leq\eps$,
\begin{align}
\langle x,a x\rangle&=\langle x,(a-\cn_\delta(\rho_s)) x\rangle +\langle x,\cn_\delta(\rho_s)x\rangle\\
&\geq\langle x,(a-\cn_\delta(\rho_s)) x\rangle+\delta/2\\
&\geq\delta/2-\eps,
\end{align}
and for $\eps\leq\delta/2$ this is larger than or equal to zero. So, applying $\cn_\delta$ guarantees us that $(\cn_\delta(\mathfrak S'))_{\delta/2}\subset\conv(B)\subset\cs(\mathbb C^d)$.\\
It can be read off from \cite{webster}, Theorems 3.1.6 and 1.8.5 and above cardinality bounds for $\delta$-nets that $|\bS'''|\leq(12/\delta)^{2d^2}$ can be chosen. Now replace $\delta$ by $\delta_n=12n^{-1/(4d^2)}$, then $|\bS'''|\leq\sqrt{n}$. Thus,
\begin{align}
\tr\{\cn_{\delta_n}^{\dag\otimes n}(P_n)\rho_{s^n}\}&\geq1-2^{-n(\alpha\eps^2+\frac{d^2+\sqrt{n}}{n}\log(2n))}\qquad\forall\rho_{s^n}\in\mathfrak S^{\otimes n}.
\end{align}
At the same time, looking back to our equations (\ref{eqn:start-of-error-exponent}) to (\ref{eqn:end-of-error-exponent}), we see that the following is valid:
\begin{align}
\tr\{\cn_\delta^{\dag\otimes n}(P_n)\sigma^{\otimes n}\}&=\tr\{P_n\cn_\delta^{\otimes n}(\sigma^{\otimes n})\}\\
&\leq2^{-n(\inf_{\rho\in\mathfrak S}D(\rho\|\cn_\delta(\sigma))-\Theta(n,\eps,d,\cn_{delta_n}(\sigma))}.
\end{align}
It still holds $\lim_{n\to\infty}\Theta(n,\eps,d,\cn_{delta_n}(\sigma))=0$, but by monotonicity of the relative entropy it is not yet clear that this is in general the optimal exponent.
By assumption, $\inf_{\rho\in\mathfrak S}D(\rho\|\sigma)<\infty$. Then also $\inf_{\rho\in\mathfrak S}D(\rho\|\cn_{\delta_n}(\sigma))<\infty$ so that for every $\rho\in\mathfrak S$ and for $\delta_n< t_d$, we also have
\begin{align}
D(\rho\|\cn_{\delta_n}(\sigma))&=D(\rho\|\sigma)-\tr\{\rho(\log(\cn_{\delta_n}(\sigma))-\log(\sigma))\}\\
&\geq D(\rho\|\sigma)-d\cdot\max_i\log(\frac{t_i}{(1-{\delta_n})t_i+{\delta_n}/2})\\
&\geq D(\rho\|\sigma)-d\cdot\log(\frac{t_1}{(1-{\delta_n})t_1+{\delta_n}/2})\\
&\geq D(\rho\|\sigma)-d\cdot\log(\frac{t_1}{(1-{\delta_n})t_1})\\
&=D(\rho\|\sigma)+d\cdot\log(1-{\delta_n}).
\end{align}
Clearly, this proves
\begin{align}
\tr\{\cn_{\delta_n}^{\dag\otimes n}(P_n)\sigma^{\otimes n}\}\leq2^{-n(\inf_{\rho\in\conv(\mathfrak S)}D(\rho\|\sigma)+d\cdot\log(1-\delta_n)-\Theta(n,d,\eps,\cn_{\delta_n}(\sigma))}.
\end{align}
%%%%%%%%%%%%%%%%%%%%%%%%%%%%%%%%%%%%%%%%%%%%%%%%%%%%%%%%%%%%%%%%%%%%%%%%%%%%%%%%%%%%%%%%%
\emph{Proof of the lower bound on the error exponent in Theorem \ref{theorem:main-result-II} for infinite $|\bS|$.}\\
%%%%%%%%%%%%%%%%%%%%%%%%%%%%%%%%%%%%%%%%%%%%%%%%%%%%%%%%%%%%%%%%%%%%%%%%%%%%%%%%%%%%%%%%%
Luckily, in this converse part of the proof we do not have to approximate $\mathfrak S$ from the outside, an approximation from the inside is enough. We will need the Fannes-Audenaert inequality that yields the following bound:
\begin{align}
|D(\rho\|\sigma)-D(\rho'\|\sigma)|\leq \tau\log(d)+h(\tau)+\tau\log(1/t_d)=:c(\tau),
\end{align}
where $\tau:=\frac{1}{2}\|\rho-\rho'\|$. Then, as before, we approximate an infinite set $\conv(\mathfrak S)$ by a sequence $\conv(\mathfrak S_n)_{n\in\nn}$ of convex sets, each $\conv(\mathfrak S_n)$ having finitely many extremal points $\{\rho_x\}^{(n)}\}_{x\in\bX_n}$, where $|\bX_n|\leq\sqrt{n}$ and for each $\rho\in\conv(\mathfrak S)$ and $n\in\nn$ there is $\rho'\in\conv(\mathfrak S_n)$ such that $\|\rho-\rho'\|\leq12 n^{-1/4d^2}$. Then, by playing everything back to the finite case and using the scaling of $|\bX_n|$ we get
\begin{align}
\tr\{A\sigma^{\otimes n}\}&\geq2^{-n(\inf_{\rho\in\conv(\mathfrak S)}D(\rho\|\sigma)+c(12 n^{-1/4d^2})+\Gamma(n,\eps,d,\sigma,\sqrt{n}))}.
\end{align}
This proves our claim, since $\lim_{n\to\infty}\Gamma'(n,\eps,d,\sqrt{n})=0$ for all $\eps\in(0,1)$.
\end{proof}
%%%%%%%%%%%%%%%%%%%%%%%%%%%%%%%%%%%%%%%%%%%%%%%%%%%%%%%%%%%%%%%%%%%%%%%%%%%%%%%%%%%%%%%%%
%%%%%%%%%%%%%%%%%%%%%%%%%%%%%%%%%%%%%%%%%%%%%%%%%%%%%%%%%%%%%%%%%%%%%%%%%%%%%%%%%%%%%%%%%
\begin{proof}[Proof of our no-go-result]
Let a permutation-invariant operator $A\in\mathcal B((\mathbb C^d)^{\otimes n})$ satisfy the inequalities $0\leq A\leq\eins$ and
\begin{align}
1-\eps\leq\tr\{A\rho^{\otimes n}\}\qquad\forall\ \rho\in\cs(\mathbb C^d)
\end{align}
for some $\eps>0$. Then due to Schur-Weyl duality, $\eins-A$ can be written $\eins-A=\sum_{\lambda}A_{\lambda}$ with each $A_\lambda$ satisfying $P_\lambda c_\lambda\geq A_\lambda\geq0$ for some $1\geq c_\lambda\geq0$, where $P_\lambda$ is the projection onto the isotypical subspace corresponding to $\lambda$. It follows that if we average an $A_\lambda$ over the unitary group, we get for every $\rho$ that
\begin{align}
\tr\{\int d_UU^{\otimes n}A_\lambda U^{\dag\otimes n}\rho^{\otimes n}\}\leq\int d_U\tr\{U^{\otimes n}(\eins-A)U^{\dag\otimes n}\rho^{\otimes n}\}\leq \int d_U\eps=\eps.
\end{align}
Let $\lambda$ be fixed for the moment. Let the invariant subspace $V_\lambda$ have a decomposition into irreducible components of $S_n$ as $V_\lambda=\bigoplus_iV_i$. The operator $A_\lambda$ can then, due to Schur-Weyl duality, be written as
\begin{align}
A_\lambda=\sum_{i,j=1}^{m_{n,\lambda}}c_{ij}Y_{ij},
\end{align}
for operators $Y_{ij}$ such that for each pair $i\neq j$, the operator $Y_{ij}$ commutes with the action of the permutation group and maps the irreducible representation $V_j$ to $V_i$ and each $P_i:=Y_{ii}$ is the orthogonal projection onto $V_i$. Moreover, $\tr\{Y_{ij}\}=\delta(i,j)\dim(F_\lambda)$. The numbers $c_{ij}$ satisfy $c_{ii}\geq0$ for all $i\in[m_{n,\lambda}]$ and $|c_{ij}|\leq \sqrt{|c_{ii}|\cdot|c_{jj}|}$, since $A_\lambda\geq0$.\\
Now choose a state $\rho$ with spectrum $t$ satisfying $t=\lambda$, and observe that $\int d_U(U\rho U^\dag)^{\otimes n}=\sum_\lambda u_\lambda P_\lambda$ for appropriate number $0\leq u_\lambda\leq 1$, so that for every $i\in[m_{\lambda,n}]$:
\begin{align}
\eps&\geq\tr\{A_\lambda\int_{U(d)}(U\rho U^\dag)^{\otimes n}\}dU\\
&=\sum_{i,j}c_{ij}\tr\{Y_{ij}\int_{U(d)}(U\rho U^\dag)^{\otimes n}\}dU\\
&=\sum_i c_{ii}\int_{U(d)}\tr\{Y_{ii}(U\rho U^\dag)^{\otimes n}\}dU\\
&=\sum_i c_{ii}\int_{U(d)}\tr\{U^{\dag\otimes n}Y_{ii}U^{\otimes n}\rho^{\otimes n}\}dU\\
&\geq c_{ii}m_{\lambda,n}^{-1}\tr\{P_\lambda\rho^{\otimes n}\}\\
&\geq c_{ii}m_{\lambda,n}^{-1}\tr\{P_{\lambda,\lambda}\rho^{\otimes n}\}\\
&\geq c_{ii}m_{\lambda,n}^{-1}\frac{\dim(F_\lambda)}{|T_\lambda|}t^{\otimes n}(T_\lambda).
\end{align}
We now need to estimate the term $\frac{\dim(F_\lambda)}{|T_\lambda|}$. With $h(i,j)$ denoting Hook-lenghts, the dimensions of the irreducible subspaces of any representation of $S_n$ on $(\mathbb C^d)^{\otimes n}$ ($d>0$) obey the following estimates.
\begin{eqnarray}
\frac{n!}{\prod_{i=1}^n(\lambda_i+d+1)!}\leq\frac{n!}{\prod_{(i,j)\in\lambda}h(i,j)}=\dim F_\lambda\leq\frac{n!}{\prod_{i=1}^d\lambda_i!}\qquad (\lambda\in \mathbb Y_{d,n}).
\end{eqnarray}
This implies that
\begin{align}
\frac{\dim(F_\lambda)}{|T_\lambda|}\geq\prod_{i=1}^d\frac{\lambda_i!}{(\lambda_i+d+1)!}\geq\prod_{i=1}^d\frac{1}{(n+d+1)^d}=(n+d+1)^{-d^2}.
\end{align}
Thus,
\begin{align}
\eps&\geq c_{ii}m_{\lambda,n}^{-1}(n+d+1)^{-d^2}t^{\otimes n}(T_\lambda)\geq c_{ii}m_{\lambda,n}^{-1}(n+d+1)^{-d^2}(2n)^{-d},
\end{align}
leading to $c_{ii}\leq (2dn)^{3d^2}\eps$ for all $i\in[m_{\lambda,n}]$. Now take any vector $x=\sum_i x_i\in V_\lambda$ such that each $x_i\in V_i$, and $\langle x,x\rangle=1$. Then
\begin{align}
\langle x,A_\lambda x\rangle&=\sum_{i,j}c_{ij}\langle x,Y_{ij}x\rangle\\
&=\sum_{i,j}c_{ij}\langle x_i,Y_{ij}x_j\rangle\\
&\leq\sum_{i,j}|c_{ij}|\\
&\leq\sum_{i,j}\sqrt{|c_{ii}|\cdot|c_{jj}|}\\
&\leq m_{n,\lambda}(2dn)^{3d^2}\eps\\
&\leq (2dn)^{4d^2}\eps,
\end{align}
and this proves that
\begin{align}
A_\lambda\leq (2dn)^{4d^2}\eps P_\lambda\qquad\forall\ \lambda\in\mathbb Y_{d,n}
\end{align}
so that we finally get the desired
\begin{align}
A\leq(1-\eps\cdot(2dn)^{4d^2})\eins.
\end{align}
\end{proof}
%%%%%%%%%%%%%%%%%%%%%%%%%%%%%%%%%%%%%%%%%%%%%%%%%%%%%%%%%%%%%%%%%%%%%%%%%%%%%%%%%%
\end{section}
%%%%%%%%%%%%%%%%%%%%%%%%%%%%%%%%%%%%%%%%%%%%%%%%%%%%%%%%%%%%%%%%%%%%%%%%%%%%%%%%%%
%%%%%%%%%%%%%%%%%%%%%%%%%%%%%%%%%%%%%%%%%%%%%%%%%%%%%%%%%%%%%%%%%%%%%%%%%%%%%%%%%%
\begin{section}{Example}
As an example, we consider $d=2$. Here, the set of states can be identified with the unit ball in $\mathbb R^3$, and this is also the only reason why we restrict ourselves to that case. Any density operator can be written in the Bloch sphere representation as $\xi=\frac{1}{2}(\eins+\sum_{i=1}^3x_i\sigma_i)$, where $\sigma_1,\sigma_2,\sigma_3\in\mathcal B(\mathbb C^2)$ are the usual Pauli matrices, $x_1,x_2,x_3\in\mathbb R$ and $\rho\in\cs(\mathbb C^2)$ is equivalent to $\sum_{i=1}^3 x_i^2\leq1$.\\
let $\sigma=\frac{1}{2}(\eins+\frac{1}{2}\sigma_3)$, and $\mathfrak S=\{\rho:x_3\leq1/4\}$. Let us assume that, by using (for some large number $n$ and some small $\eps>0$), the POVM $\mathcal M:=\{\eins-P_n\}\cup\{P_{f,\lambda}\}_{(f,\lambda)\in\Lambda_\eps}$. The measurement outcome associated to $\eins-P_n$ shall be denoted $e$. But then for every $\rho_s\in\mathfrak S$,
\begin{align}
\tr\{P_{n,s}\rho_s^{\otimes n}\}\geq1-(2n)^82^{-n\ln(2)\eps^2}.
\end{align}
This implies the following. Assume that a state $\xi\in\mathfrak S\cup\{\sigma\}$ is being sent, and that we try to find out which state it was. Let  $\tilde r_\rho$ denote the distribution corresponding to the eigenvalues of the pinching of $\rho$ to the eigenbasis of $\sigma$ and $r_\rho$ the spectrum of $\rho$. Define the sets
\begin{align}
B(f,\lambda):=\{\rho\in\mathfrak S\cup\{\sigma\}:\|\bar f-\tilde r_\rho\|\leq\eps\ \bigwedge\ \|\bar\lambda-r_\rho\|\leq\eps\}.
\end{align}
After application of our measurement $\mathcal M$ to the system, we know the following:
\begin{align}
\mathbb P(\xi=\sigma|\mathcal M=e)&\geq1-2^{-n\min_{\rho\in\mathfrak S}D(\sigma\|\rho)-\Theta(n,\eps))}\\
\mathbb P(\xi\in B(f,\lambda)|\mathcal M=(f,\lambda))&\geq1-(2n)^82^{-n\ln(2)\eps^2}.
\end{align}
So, in the end, we can (asymptotically, with high probability) not only distinguish $\sigma$ from $\mathfrak S$, but to some extent also the elements of $\mathfrak S$ from one another (with accuracy $\eps$, and a drawback being that states with the same spectrum and the same pinching to the eigenbasis of $\sigma$ cannot be distinguished). In our example for $d=2$ this means that we get the $z$-coordinate of the true state, and its spectrum (=distance from the origin).
\end{section}
\\\\
%%%%%%%%%%%%%%%%%%%%%%%%%%%%%%%%%%%%%%%%%%%%%%%%%%%%%%%%%%%%%%%%%%%%%%%%%%%%%%%%%%%%
%%%%%%%%%%%%%%%%%%%%%%%%%%%%%%%%%%%%%%%%%%%%%%%%%%%%%%%%%%%%%%%%%%%%%%%%%%%%%%%%%%%%
\emph{Conclusion.}\\
It is an open question whether an even finer decomposition of $V_f$ into a more subtle choice of irreducible subspaces may reveal even more information about the states being measured. A look at pure permutation invariant von Neumann POVMs on the symmetric subspace suggests that there will always be a small ambiguity left.\\
However, switching to non-optimal (in the sense of Stein's Lemma) measurements that are not permutation-invariant any more (e.g. by concatenating the measurements that are known to be optimal for state discrimination in the sense of Stein's Lemma) might lead to sub-optimal, but information complete measurements.\\
The proof of Theorem \ref{theorem:additional-result} suggests that the 'quantum method of types' still needs more refinement: For each $V_{f,\lambda}$, the multiplicity of $F_\lambda$ in $V_f$ can be larger than one if $d>2$. This suggests that there is still some information contained in the systems we consider that is not exploited by our approach.\\
This is in no contradiction to what we said earlier in our example.
%-------------------------------------------------------------------------------
%-----------------------------------------------------------------------------------
\ \\\\
\emph{Acknowledgement.}
This work was supported by the BMBF via grant 01BQ1050, the hospitality of the Isaac Newton Institute for Mathematical Sciences in the final stages of preparation of this manuscript is gratefully acknowledged. Many thanks go to Koenraad Audenaert, Fernando Brandao, Matthias Christandl, David Gross, Piotr \'{S}niady and Michael Walter for stimulating discussions.
%-----------------------------------------------------------------------------------


\begin{thebibliography}{99}
\bibitem{ahlswede-avc-with-state-seq-known-to-sender} R. Ahlswede, ``Arbitrarily Varying Channels with States Sequence Known to the Sender'', \emph{IEEE Trans. Inf. Th.} Vol. 32, 621-629 (1986)

\bibitem{audenaert-nussbaum-szkola-verstraete} K.M.R Audenaert, M. Nussbaum, A. Szkola, F. Verstraete ``Asymptotic Error Rates in Quantum Hypothesis Testing'', \emph{Comm. Math. Phys.} Vol. 279, 251-283 (2008)

\bibitem{fannes-audenaert} K.M.R. Audenaert ``A sharp continuity estimate for the von Neumann entropy'', \emph{J. Phys. A: Math. Theor.} Vol. 40, 8127 (2007)\\
M. Fannes, ``A continuity property of the entropy density for spin lattice systems'', \emph{Comm. Math. Phys.} Vol. 31, 291-294 (1973).

\bibitem{bacon-chuang-harrow} D. Bacon, I.L. Chuang, A.W. Harrow, ``The quantum Schur and Clebsch-Gordan transforms: I. efficient qudit circuits'', \emph{Proc. SODA '07, eighteenth annual ACM-SIAM symposium on Discrete algorithms} 1235-1244 (2007)

\bibitem{bhatia-spectral-variation} R. Bhatia, ``Analysis of spectral variation and some inequalities'', \emph{Trans. Amer. Math. Soc.} Vol. 272, 323-331 (1982)

\bibitem{bbn1}I. Bjelakovi\'c, H. Boche, J. N\"otzel, ``Quantum capacity of a class of compound channels'', \emph{Phys. Rev. A} 78, 042331, (2008)

\bibitem{bjelakovic-siegmund_schultze} I. Bjelakovi\c{c}, R. Siegmund-Schultze, ``Quantum Stein’s lemma revisited, inequalities for quantum entropies, and a concavity theorem of Lieb'', \emph{arXiv:quant-ph/0307170} (2003), revised and extended version (2012)

\bibitem{bjelakovic-siegmund_schultze-ergodic} I. Bjelakovi\c{c}, R. Siegmund-Schultze, ``An Ergodic Theorem for the Quantum Relative Entropy'', \emph{Comm. Math. Phys.} vol. 247, 697-712 (2004)

\bibitem{bjelakovic-deuschel-krueger-seiler-siegmund_schultze-szkola-sanov} I. Bjelakovi\c{c}, J.D. Deuschel, T. Kr\"uger, R. Seiler, R. Siegmund-Schultze, A. Szkola, ``A Quantum Version of Sanov's Theorem'', \emph{Comm. Math. Phys.} Vol. 260, Iss. 3, 659-671 (2005)

\bibitem{bjelakovic-deuschel-krueger-seiler-siegmund_schultze-szkola-typical-support-and-sanov} I. Bjelakovi\c{c}, J.D. Deuschel, T. Kr\"uger, R. Seiler, R. Siegmund-Schultze, A. Szkola, ``Typical Support and Sanov Large Deviations of Correlated States``, \emph{Comm. Math. Phys.} Vol. 279, Iss. 2, 559-584 (2008)

\bibitem{brandao-harrow-lee-peres} F.G.S.L. Brand\~{a}o, A.W. Harrow, J.R. Lee, Y. Peres, ``Adversarial Hypothesis Testing and Quantum Stein's Lemma for Restricted Measurements'', \emph{arXiv:1308.6702} (2013)

\bibitem{brandao-plenio} F.G.S.L. Brand\~{a}o, M.B. Plenio, ``A generalization of quantum Stein’s lemma'', \emph{Comm. Math. Phys.} Vol. 295, No. 3, 791-828 (2010)

\bibitem{christandl-thesis} M. Christandl, ``The Structure of Bipartite Quantum States - Insights from Group Theory and Cryptography'', thesis, (2006)

\bibitem{christandl-mitchison}
M. Christandl and G. Mitchison, ``The spectra of density operators and the Kronecker coefficients of the symmetric group'',
\emph{Comm. Math. Phys.}, Vol. 261, Issue 3, 789-797 (2006)

\bibitem{christandl-harrow} M. Christandl, A. Harrow, G. Mitchison, ``Nonzero Kronecker Coefficients and What They Tell us about Spectra'', \emph{Comm. Math. Phys.} Vol. 270, No. 3, 575-585 (2007)

\bibitem{christandl-horn_and_LW_coefficients} M. Christandl, ``A quantum information-theoretic proof of the relation between Horn's problem and the Littlewood-Richardson coefficients'', \emph{Proc. CiE 2008, LNCS}, Vol. 5028, 120-128 (2008)

\bibitem{christandl-sahinoglu-walter} M. Christandl, M.B. Sahinoglu, M. Walter, ``Recoupling Coefficients and Quantum Entropies'', \emph{ 	arXiv:1210.0463} (2012)

\bibitem{walter-doran-gross-christandl} M. Walter, B. Doran, D. Gross, M. Christandl, ``Entanglement Polytopes'', \emph{Science}, Vol. 340, No. 6137, 1205-1208 (2013)

\bibitem{csiszar-types} I. Csisz\'ar, ``The Method of Types'', \emph{IEEE Trans. Inf. Theory}, Vol. 44, No. 6 (1998)

\bibitem{csiszar-koerner}
I. Csiszar, J. K\"orner, \emph{Information Theory; Coding Theorems for Discrete Memoryless Systems}, Akad\'emiai Kiad\'o, Budapest/Academic Press Inc., New York 1981

\bibitem{fangwei-shiyi} F. Fangwei, S. Shiyi, ``Hypothesis testing for arbitrarily varying source'', \emph{Acta Mathematica Sinica}, Vol. 12, No. 1, 33-39 (1996)

\bibitem{fulton} W. Fulton, ``Young Tableaux With Applications to Representation Theory and Geometry'', \emph{Cambridge University Press} (1997)

\bibitem{hardy-littlewood-polya} G.H. Hardy, J.E. Littlewood, G. P\'oly, ``Inequalities'', \emph{Cambridge Mathematical Library} (1934)

\bibitem{harrow-diss}
A. Harrow, ``Applications of coherent classical communication and the Schur transform to quantum information theory'' \emph{PhD-thesis}, available at http://arxiv.org/abs/quant-ph/0512255v1 (2005)

\bibitem{hayashi-representation-theory} M. Hayashi, ``Asymptotics of quantum relative entropy from a representation theoretical viewpoint'', \emph{J. Phys. A: Math. Gen.} Vol. 34, 3413-3419 (2001)

\bibitem{hayashi-nagaoka} T. Ogawa, M. Hayashi, ``On error exponents in quantum hypothesis testing'', \emph{IEEE Trans. Inf. Theory}, Vol. 50, No. 6, 1368 - 1372 (2004)

\bibitem{hayashi01} M. Hayashi, ``Optimal sequence of POVMs in the sense of Stein's lemma in quantum hypothesis testing'', \emph{arXiv:quant-ph/0107004} (2001)

\bibitem{hayashi02} M. Hayashi, ``Optimal sequence of quantum measurements in the sense of Stein's lemma in quantum hypothesis testing'', \emph{J. Phys. A Math. Gen.} Vol. 35, 10759 - 10773 (2002)

\bibitem{hayashi-matsumoto}
M. Hayashi and K. Matsumoto, ``Quantum universal variable-length source coding'', \emph{Phys. Rev. A}, 66(2):022311, (2002)

\bibitem{hiai-mosonyi-ogawa} F. Hiai, M. Mosonyi, T. Ogawa, ``Error exponents in hypothesis testing for correlated states on a spin chain'', \emph{J. Math. Phys.} Vol. 49, 032112 (2008)

\bibitem{hiai-petz} F. Hiai, D. Petz, ``The proper formula for relative entropy and its asymptotics in quantum probability'', \emph{Comm. Math. Phys.},  Vol. 143, Iss. 1, 99-114 (1991)

\bibitem{horn-lemma} A. Horn, ``Doubly Stochastic Matrices and the Diagonal of a Rotation Matrix'', \emph{Amer. J. Math.}, Vol. 76, No. 3, 620-630 (1954)

\bibitem{keyl-werner}
M. Keyl, R. F. Werner. ``Estimating the spectrum of a density operator'', \emph{Phys. Rev. A}, 64(5):052311, (2001)

\bibitem{keyl} M. Keyl, ``Quantum State Estimation and Large Deviations'', \emph{Rev. Math. Phys.} Vol. 18, No. 1, 19-60 (2006).

\bibitem{milman-schechtman} V.D. Milman, G. Schechtman, ``Asymptotic Theory of Finite Dimensional Normed Spaces'', \emph{Lecture Notes in Mathematics} 1200 (Springer-Verlag, Berlin, corrected second printing, 2001)

\bibitem{noetzel-2ptypicality} J. N\"otzel, ``A solution to two party typicality using representation theory of the symmetric group'',  	arXiv:1209.5094 (2012)

\bibitem{ogawa-nagaoka} T. Ogawa, H. Nagaoka, ``Strong Converse and Stein's Lemma in Quantum Hypothesis Testing'', \emph{IEE Trans. Inf. Theory}, Vol. 46, Iss. 7, 2428 - 2433 (2000)

\bibitem{nagaoka} H. Nagaoka, ``Strong converse theorems in quantum information theory'', \emph{Proceedings of Symposium ``ERATO Workshop on Quantum Information Science 2001''}, 33 (2001)

\bibitem{rudnicki-alicki-sadowsky}
S. Rudnicki, R. Alicki and S. Sadowski, ``Symmetry properties of product states for the system of N n-level atoms'', \emph{J. Math. Phys.}, 29(5):1158–1162, (1988)

\bibitem{schilling-christandl-gross} C. Schilling, D. Gross, M. Christandl, ``Pinning of Fermionic Occupation Numbers'', \emph{Phys. Rev. Lett.} Vol. 110, 040404 (2013)

\bibitem{sternberg} S. Sternberg, ``Group Theory and Physics'', \emph{Cambridge University Press} (1994)

\bibitem{thoma} E. Thoma, ``Die unzerlegbaren, positiv-definiten Klassenfunktionen der abz\"ahlbar unendlichen, symmetrischen Gruppe'', \emph{Math. Z.} Vol. 85, 40-61 (1964)

\bibitem{webster} R. Webster, \emph{Convexity}, Oxford University Press 1994

\bibitem{thm2}
The result we refer to has a long history that is explained in more detail in \cite{christandl-thesis}. According to
\cite{christandl-thesis}, it first appeared in \cite{rudnicki-alicki-sadowsky}, was then independently proven in \cite{keyl-werner}, later appeared in \cite{hayashi-matsumoto} with a shortened proof, and was, at last, restatet in \cite{christandl-mitchison}, again with the (obtainded independently from \cite{hayashi-matsumoto}) shortened version of the proof.\\
We would like to add to this history that it also follows directly from the much earlier work \cite{thoma}.
\end{thebibliography}
\end{document}